\documentclass[a4paper,11pt]{article}

\usepackage{fullpage}
\usepackage{authblk}

\usepackage[utf8]{inputenc} % allow utf-8 input
\usepackage[T1]{fontenc}    % use 8-bit T1 fonts
\usepackage{hyperref}       % hyperlinks
\usepackage{url}            % simple URL typesetting
\usepackage{booktabs}       % professional-quality tables
\usepackage{amsfonts}       % blackboard math symbols
\usepackage{nicefrac}       % compact symbols for 1/2, etc.
\usepackage{microtype}      % microtypography
\usepackage{xcolor}         % colors

\hypersetup{
	colorlinks   = true,
    linkcolor    = black, % color of internal links
	urlcolor     = black, % color of external links
	citecolor    = blue  % color of links to bibliography
}

\usepackage{algorithm}
\usepackage{algorithmic}
\usepackage{bm, bbm}
\usepackage{amssymb,amsmath,amsfonts,amstext,amsthm, mathtools}
\usepackage{cleveref}
\Crefname{equation}{}{}

\usepackage{natbib}
\usepackage[inline]{enumitem}
\usepackage{booktabs, makecell}

\renewcommand{\cite}[1]{\citep{#1}}

\usepackage{pgfplots}
\usepackage{tikz}
\usetikzlibrary{intersections,positioning,calc,shapes.misc}

\usepackage{breakcites}

\usepackage{cleveref}
\usepackage{thm-restate}

\usepackage{framed}

\newcommand{\vv}{\mathbf{v}}
\usepackage{bm}

\newcommand{\va}{\mathbf{a}}
\newcommand{\bb}{\mathbf{b}}

\newcommand{\rr}{\mathbf{r}}

\renewcommand{\vv}{\mathbf{v}}

\newcommand{\xx}{\mathbf{x}}
\newcommand{\yy}{\mathbf{y}}
\newcommand{\zz}{\mathbf{z}}

\newcommand{\calB}{\mathcal{B}}

\newcommand{\calM}{\mathcal{M}}

\newcommand{\calO}{\mathcal{O}}
\newcommand{\calP}{\mathcal{P}}

\newcommand{\calV}{\mathcal{V}}
\newcommand{\calW}{\mathcal{W}}

\newcommand{\tda}{\tilde{a}}

\newcommand{\wt}{\widetilde}
\newcommand{\wh}{\widehat}

\newcommand{\whp}{\widehat{p}}

\newcommand{\whV}{\widehat{V}}

\newtheorem{lemma}{Lemma}

\newtheorem{claim}{Claim}

\theoremstyle{definition}
\newtheorem{definition}[lemma]{Definition}
\newtheorem{example_}[lemma]{Example}

\newcommand{\hltb}[1]{{\color{blue}#1}}

\newcommand{\reg}{\mathrm{Reg}}
\newcommand{\abs}[1]{\left| #1 \right|}
\DeclareMathOperator*{\argmax}{arg\,max}
\DeclareMathOperator*{\argmin}{arg\,min}

\newcommand{\E}{\mathbb{E}}
\newcommand{\given}{{\,|\,}}

\newenvironment{talign*}
 {\csname align*\endcsname}
 {\endalign}

\newcommand{\sfP}{\mathsf{P}}
\newcommand{\sfA}{\mathsf{A}}

\newcommand{\poly}{\operatorname{poly}}
\newcommand{\vomega}{\bm{\omega}}

\newcommand{\wtomega}{\widetilde{\omega}}
\renewcommand{\varpi}{\bar{\pi}}

\title{Stochastic Principal-Agent Problems:\\ Efficient Computation and Learning}

\author[1]{Jiarui Gan}
\author[2]{Rupak Majumdar}
\author[3]{Debmalya Mandal}
\author[2]{Goran Radanovic}

\affil[1]{University of Oxford}
\affil[2]{Max Planck Institute for Software Systems}
\affil[3]{University of Warwick}

\date{}

\begin{document}

\maketitle

\begin{abstract}
We introduce a stochastic principal-agent model. A principal and an agent interact in a stochastic environment, each privy to observations about the state not available to the other. The principal has the power of commitment, both to elicit information from the agent and to provide signals about her own information. The players communicate with each other and then select actions independently. Each of them receives a payoff based on the state and their joint action, and the environment transitions to a new state. The interaction continues over a finite time horizon. Both players are {\em far-sighted}, aiming to maximize their total payoffs over the time horizon. The model encompasses as special cases extensive-form games (EFGs) and stochastic games of incomplete information, partially observable Markov decision processes (POMDPs), as well as other forms of sequential principal-agent interactions, including Bayesian persuasion and automated mechanism design problems. 

We consider both the computation and learning of the principal's optimal policy. Since the general problem, which subsumes POMDPs, is intractable, we explore algorithmic solutions under {\em hindsight observability}, where the state and the interaction history are revealed at the end of each time step.
Though the problem becomes more amenable under this condition, the number of possible histories remains exponential in the length of the time horizon, making approaches for EFG-based models infeasible. We present an efficient algorithm based on constructing the {\em inducible value sets}. The algorithm computes an $\epsilon$-approximate optimal policy in time polynomial in $1/\epsilon$. 
Additionally, we show an efficient learning algorithm for a typical episodic reinforcement learning setting where the transition probabilities are unknown. The algorithm guarantees sublinear regret $\widetilde{\mathcal{O}}(T^{2/3})$ for both players over $T$ episodes.
\end{abstract}

\section{Introduction}

Many problems in economic theory involve sequential reasoning between multiple parties with asymmetric access to information \cite{ross1973economic,JM76,BM04,LS18}.
For example, in contract theory, one party (the principal) delegates authority and decision-making power to another (the agent), and the goal is to design mechanisms to ensure that the agent's actions align with the principal's utilities.
This broad class of \emph{principal-agent problems} lead to many research questions about information design and optimal strategic behaviors, with broad-ranging applications from governance and public administration to e-commerce and financial services. 
In particular, \emph{algorithmic} techniques for optimal decision making and learning are crucial for obtaining effective solutions to real-world problems in this domain.

In this paper, we consider a general framework for stochastic principal-agent problems.
We study algorithmic problems related to the computation and learning of optimal solutions under this framework. 
In this framework, the interaction between the principal and the agent takes place in a stochastic environment over multiple time steps.
In each step, both players are privy to information not available to the other and make partial observations about the environment. 
The players can communicate their private information to influence each other and, based on this communication, play actions that jointly influence the state of the environment.
Each player has their own payoff, and we study the \emph{general sum case} where the payoffs need not sum to zero.
The players are \emph{far-sighted}: their goal is to maximize their expected total payoffs over the entire horizon of the game.
Technically, these are stochastic games with partial information on both sides \cite{AumannMaschler,MertensSolanZamir}.

In line with the principal-agent framework, we assume that the principal has the power of \emph{commitment}, both to elicit information from the agent and to provide signals about her own information to coordinate their joint actions. 
A commitment is a binding agreement for the principal to act according to the committed strategy; technically, we have a
Stackelberg game \cite{von2010market}.
The agent acts optimally in response to the commitment, deciding what information to share and what actions to perform at their own discretion. 
As a result, our model incorporates both sequential Bayesian persuasion (or information design \cite{kamenica2011bayesian}) \cite{gan2022sequential,wu2022sequential} and sequential mechanism design \cite{zhang2021automated}
as special cases, as well as extensive-form games (EFGs) and stochastic games with coordinated communication, and partially observation Markov decision processes (POMDPs). 
The model is strictly more expressive than EFG-based models of similar principle-agent problems as history sequences are represented more concisely through a Markov process. The number of possible histories may therefore grow exponentially as the length of the time horizon increases, while in EFGs this is normally bounded by the input size (i.e., the size of the game tree). For this reason, we coin the term {\em stochastic principal-agent problem}, for the similarity of this model to {\em stochastic games} \cite{shapley1953stochastic}.

We focus on a finite time horizon and the total payoff.
We consider both the \emph{full information} setting, where all parameters of the underlying game are known to both players and their goal is to design optimal policies, and the \emph{partial information} setting, where the parameters are not given beforehand and have to be learned by interacting in the environment.
Based on these two settings, we design algorithms to compute or learn the principal's optimal policy, which is in general history-dependent.

\subsection{Our Results}

Since the general setting of our model subsumes POMDPs---which are PSPACE-hard when the horizon is finite \cite{papadimitriou1987complexity}---we explore a \emph{hindsight observability} condition in the literature on POMDPs \cite{lee2023learning}, whereby the hidden interaction history is revealed to both players at the end of each time step.
Under this condition, our first main result is an efficient near-optimal algorithm for computing the principal's optimal policy. The algorithm is based on a dynamic programming approach which works by constructing the {\em inducible value sets}. The algorithm computes an $\epsilon$-approximate optimal policy, that is optimal up to any desired (additive) approximation error $\epsilon$, in time polynomial in $1/\epsilon$.
The key technical difficulty in designing the algorithm is to characterize the one-step solutions in the dynamic programming formulation, as projections of convex polytopes that can be efficiently approximated up to an additive error.

Next, we study the partial-information case. 
We consider a typical reinforcement learning (RL) setting where the transition model is not given beforehand and needs to be learned by interacting with the environment. The setting is episodic and consists of $T$ episodes. 
As our second main result, we present a learning algorithm that guarantees sublinear $\wt{\calO}(\poly(M,H) \cdot T^{2/3})$ regret for both players, where $M$ is the size of the model and $H$ is the horizon length of each episode. The bound matches a $\Omega(T^{2/3})$ lower bound presented in previous work for a sequential persuasion model \cite{bernasconi2022sequential}.
Our learning algorithm uses {\em reward-free exploration} from the recent RL literature, and relies on efficient computation of optimal policies that are {\em approximately} incentive compatible. The latter is achieved via a variant of our algorithm for the full-information case.

\subsection{Related Work}

The principal-agent problem is a well-known concept in economics studies~\citep[see, e.g.,][]{ross1973economic,myerson1982optimal,MR86,makris2003theory}.
Models featuring sequential interactions have also been proposed and studied \cite{myerson1986multistage,forges1986approach}. 
Our work follows the same modeling approach as these early works and generalizes the one-shot versions of the respective types of principal-agent problems, including information design \cite{kamenica2011bayesian}, automated mechanism design \cite{sandholm2003automated}, as well as mixtures of the two \cite{myerson1982optimal,bayesian2022castiglioni,gan2022optimal}.
In the more recent literature, there has been a growing interest in the algorithmic aspects of these sequential models. The computation and learning of sequential extensions of various forms of principal-agent problems have been studied (e.g., information design \cite{celli2020private,gan2022bayesian,gan2022sequential,wu2022sequential,bernasconi2022sequential}, automated mechanism design \cite{zhang2021automated,cacciamani2023online}, other types of sequential Stackelberg games \cite{letchford2010computing,letchford2012computing,bovsansky2017computation,harris2021stateful,collina2023efficient}, and even more recently,  contract design \cite{ivanov2024principal}).

Our model can be viewed as a generalization of the above works, incorporating a stochastic setting with a finite horizon and far-sighted players. Specifically, \citet{gan2022bayesian} first introduced an infinite-horizon information design model based on an MDP. They showed that optimal {\em stationary} strategies are inapproximable, unless the receiver is myopic. This work left open the tractability of optimal {\em history-dependent} strategies, especially in finite-horizon models, which we consider in this paper. 
\citet{wu2022sequential} later studied the reinforcement learning problem against a myopic agent in the same sequential information design model. 
\citet{bernasconi2022sequential} also studied a model based on an EFG and presented efficient computation and learning algorithms. Similar EFG-based models have also been explored in the recent literature \cite{zhang2022polynomial,zhang2024computing}. 
EFGs are less expressive than MDP-based models since possible history sequences are explicitly given in the model. The number of possible histories is bounded by the size of the problem as a result, where as this can be exponential in an MDP. 
Hence, efficient algorithms for EFG-based models do not directly translate to efficient algorithms for our MDP-based model. 
In the domain of automated mechanism design, \citet{zhang2021automated} studied a finite-horizon model that is a POMDP for the principal and MDP for the agent. They presented an LP (linear program) for computing  optimal mechanisms, though the size of the LP is exponential in the size of the problem. 

Our algorithm for computing optimal history-dependent strategies leverages the technique of approximating inducible value sets using convex polytopes. Similar techniques have been proposed in earlier works by \citet{dermed2009solving} and \citet{macdermed2011quick} to compute optimal correlated equilibria of stochastic games. We extend these techniques into the principal-agent setting, with adaptions that ensure exact incentive compatibility (IC) in the full-information setting. 
In a closely-related work concurrent to ours, \citet{bernasconi2024persuading} used a similar approximation approach to solve an information design problem (as a special case of our model). Compared to their results, our algorithm guarantees {\em exact} IC, with a simpler approach. Moreover, we also study the learning setting, in addition to the full-information computation problem they focused on.
We note that while all the above works (including ours) only guarantee near-optimality, exact solutions are possible in some settings. In a recent work, \citet{zhang2023efficiently} presented a sophisticated exact algorithm for computing optimal correlated equilibria in two-player turn-based stochastic games.

\section{Preliminaries}
\label{sec:prelim}

A principal ($\sfP$) and an agent ($\sfA$) interact in a finite-horizon POMDP $\calM = \langle S, A, \Omega, p, \rr \rangle$, where:
$S$ is a finite state space; 
$A = A^\sfP \times A^\sfA$ is a finite joint action space;
$\Omega = \Omega^\sfP \times \Omega^\sfA$ is a finite joint observation space;
$p = (p_h)_{h=0}^{H-1}$ and $\rr = (\rr_h)_{h=1}^H$ are two tuples, each consisting of an element for every time step $h$.
Specifically, $p_0 \in \Delta(S\times \Omega)$ is a distribution of the initial state-observation pairs, and each $p_h$, $h\ge1$, is a transition function $p_h : S \times A \to \Delta(S\times \Omega)$.
Each $\rr_h = (r_h^\sfP, r_h^\sfA)$ is a pair of reward functions $r_h^\sfP: S \times A \to \mathbb{R}$ and $r_h^\sfA: S \times A \to \mathbb{R}$, for the principal and the agent, respectively.
W.l.o.g., we assume that all rewards are in $[0,1]$, and all rewards generated in the last time step $H$ are $0$.

The interaction proceeds as follows.
At the beginning, an initial state-observation pair $(s_1,\vomega_1) \sim p_0$ is drawn. 
Then, each time step $h = 1,\dots, H$ involves the following stages.
\begin{enumerate}%[leftmargin=5mm,itemsep=1mm]
\item {\bf Observation:}
The principal and the agent observe, privately, $\omega_h^\sfP$ and $\omega_h^\sfA$, respectively (but not $s_h$).

\item {\bf Communication:}
The principal elicits the agent's observation. The agent reports $\wtomega_h^\sfA \in \Omega^\sfA$ (possibly different from $\omega_h^\sfA$). Then, based on $\omega_h^\sfP$ and $\wtomega^\sfA$ the principal sends a coordination signal $a_h^\sfA$, which as we will demonstrate is w.l.o.g. an action she recommends the agent to play. 
The agent observes the recommendation $a_h^\sfA$.

\item {\bf Action:}
Based on the information exchange above, the principal and the agent, simultaneously, each perform an action, say $a_h^\sfP$ and $\tda_h^\sfA$, respectively. (The action $\tda_h^\sfA$ the agent actually performs may be different from the recommended one $a_h^\sfA$.)

\item {\bf Rewards and next state:}
Rewards $r_h^\sfP(s_h, a_h^\sfP, \tda_h^\sfA)$ and $r_h^\sfA(s_h, a_h^\sfP, \tda_h^\sfA)$ are generated for the principal and agent, respectively.\footnote{To ease the notation, we sometimes write a joint action (or observation) as two separate elements instead of a tuple. We also use commas and semicolons interchangeably as separators in a tuple, where semicolons are mainly used for differentiating elements belonging to different time steps.}
The next state is drawn: $s_{h+1} \sim p_h(\cdot \given s_h, a_h^\sfP, \tda_h^\sfA)$. 
\end{enumerate}
The model generalizes several types of principal-agent interaction, including information design (where the principal is the observer and the agent acts), automated mechanism design (where the agent is the observer and the principal acts), and stochastic games with commitment and coordination (where the environment is fully observable).

Following the general paradigm of principal-agent problems, we consider the principal's {\em commitment} to a coordination policy. The agent best-responds to the principal's commitment.
Both players are {\em far-sighted} and aim to maximize their total reward obtained over the $H$ time steps.\footnote{While we do not assume reward discounting, all our results can be easily extended to capture discounted rewards.} 
We take the principal's perspective and the goal, as we will shortly formalize, is to compute the principal's optimal commitment.
At a high level, this is a Stackelberg game between the principal and the agent and we aim to compute a Stackelberg equilibrium.

\subsection{Hindsight Observability}

Unsurprisingly, the model we have described so far is in general intractable because it generalizes POMDPs. Solving POMDPs is known to be PSPACE-hard~\cite{papadimitriou1987complexity}. The hardness remains even in the above-mentioned special cases of the model. Given this complexity barrier, we focus on the setting with {\em hindsight observability}, following the literature on POMDPs \cite{lee2023learning}.\footnote{This simplifies a conditional independence assumption in a previous preprint version \cite{gan2023sequentialv2}.}
Under hindsight observability, the interaction history is revealed to the players at the end of each time step (or equivalently, it is encoded in the players' observations in the next time step). 

It is essential that both players observe the history in hindsight. Otherwise, the model remains a generalization of POMDPs and PSAPCE-hard, even when the principal observes everything throughout (see a discussion in \Cref{sec:hardness}).
Although hindsight observability may limit some generality, the model remains quite expressive and covers a range of important scenarios, including: scenarios where the state is immediately observable, e.g., repeated games, stochastic games with full state observability \cite{collina2023efficient}, as well as scenarios where observations can be interpreted as external parameters generated based on an internal Markovian state observable to both players (e.g., \cite{gan2022bayesian,wu2022sequential}).

\subsection{History-dependent Policy}
\label{sc:history-dependent-strategies}

We consider history-dependent policies, which are more general than stationary policies and hence typically yield higher payoffs. 
For example, to punish the agent for performing a certain action requires a history-dependent policy that remember the agent's action in the previous time step.
History-dependent policies are also a natural choice for finite-horizon models, like the one we consider, where the memory required to track the history is bounded by the horizon length.

A history up to time step $h$ is a sequence 
$\sigma = \left(s_\ell, \vomega_\ell, \wtomega_\ell^\sfA, \va_\ell, \tda_\ell^\sfA \right)_{\ell=1}^h$, containing elements in the four stages of each step described above (and we write $\vomega_\ell = (\omega_\ell^\sfP, \omega_\ell^\sfA)$ and $\va_\ell = (a_\ell^\sfP, a_\ell^\sfA)$).
We let $\Sigma_h$ denote the set of all sequences till time step $h$, and let $\Sigma = \bigcup_{h=0}^H \Sigma_h$, where $\Sigma_0 = \{\varnothing\}$ contain only the empty sequence. 
Moreover, we denote by $\bar{\Sigma} \coloneqq S \times \Omega \times \Omega^\sfA \times A \times A^\sfA$ the set of all possible interactions within one time step. 
We can now write the transition function as $p_h(\cdot \given \sigma) = p_h(\cdot \given s_h, \va_h)$ for any given sequence $\sigma \in \Sigma_h$ (specially, $p_0(\cdot \given \varnothing) = p_0(\cdot)$).

\paragraph{Principal's Policy}

A history-dependent policy takes the form $\pi: \Sigma \times \Omega \to \Delta(A)$, whereby upon seeing $\sigma$ in the previous steps, observing $\omega^\sfP$, and receiving the agent's report $\wtomega^\sfA$ in the current step, the principal draws a joint action $\va = (a^\sfP, a^\sfA) \sim \pi(\sigma; \omega^\sfP, \wtomega^\sfA)$, sends $a^\sfA$ to the agent as an action recommendation, and performs $a^\sfP$ herself. 
We denote by $\pi(\va \given \sigma; \omega^\sfP, \wtomega^\sfA)$ the probability of each $\va$ in $\pi(\sigma; \omega^\sfP, \wtomega^\sfA)$.

\paragraph{Agent's Response}

The principal's commitment results in a meta-POMDP for the agent. The agent reacts by playing optimally in this meta-POMDP. When the principal's policy is IC, this simply means responding truthfully. 
Formally, the agent's strategy can be described by a {\em deviation plan} $\rho: (\sigma, \omega^\sfA) \mapsto (\wtomega^\sfA, f: A^\sfA \to A^\sfA)$, such that given any history $\sigma$ and observation $\omega^\sfA$ in the current step, the agent reports $\wtomega^\sfA$ and then plays $\tda^\sfA = f(a^\sfA)$ if subsequently the principal recommends playing $a^\sfA$. For simplicity, we write $\wtomega^\sfA = \rho(\sigma; \omega^\sfA)$ and $\tda^\sfA = \rho(\sigma; \omega^\sfA, a^\sfA)$.
We denote by $\perp$ the special deviation plan where no deviation is made, i.e., $\perp(\sigma; \omega^\sfA) \equiv \omega^\sfA$ and $\perp(\sigma; \omega^\sfA, a^\sfA) \equiv a^\sfA$.

The agent's value (i.e., total reward) induced by a policy $\pi$ and a deviation strategy $\rho$ can be defined recursively via the value function as follows. For every $h=1,\dots,H-1$ and $\sigma \in \Sigma_{h-1}$: 
\begin{align}
\label{eq:V}
&V_h^{\sfA,\pi,\rho} \left(\sigma \right) \coloneqq
\E_{(s, \vomega) \sim p_{h-1}(\cdot \given \sigma)}\
\E_{\va \sim \pi( \cdot \given \sigma, \omega^\sfP, \wtomega^\sfA)}
\left( r_h^\sfA \left(s, a^\sfP,\tda^\sfA \right) + 
V_{h+1}^{\sfA, \pi, \rho} \left( \sigma;\, s, \vomega, \wtomega^\sfA, \va, \tda^\sfA \right) \right),
\end{align}
where $\wtomega^\sfA = \rho(\sigma^\sfA, \omega^\sfA)$ and $\tda^\sfA = \rho(\sigma^\sfA, \omega^\sfA, a^\sfA)$, and by assumption $V_H^\sfA(\sigma) \equiv 0$ for the last time step.
The principal's value is defined the same way by changing the labels.

Our goal is to find a policy $\pi$ that maximizes the principal's value under the agent's best response:
\begin{align}
\label{eq:main-prob}
\max_{\pi, \rho} & \quad V_1^{\sfP,\pi,\rho}(\varnothing) \\
\text{subject to} & \quad
\rho \in \argmax\nolimits_{\rho'} V_1^{\sfA,\pi,\rho'}(\varnothing)
\label{eq:main-prob-const-1} \tag{\ref{eq:main-prob}-1}
\end{align}
In other words, we look for $\pi$ and $\rho$ that form a Stackelberg equilibrium.
We say that policy $\pi$ is {\em $\epsilon$-optimal} if $V_1^{\sfP,\pi,\rho}(\varnothing) \ge V^* - \epsilon$ for some $\rho$ satisfying \Cref{eq:main-prob-const-1}, where $V^*$ denotes the optimal value of \Cref{eq:main-prob}.

As we will demonstrate, under hindsight observability, it is without loss of optimality to consider policies that are IC (incentive compatible), which incentivize $\perp$ as an optimal response of the agent.

\begin{definition}[IC policy]
\label{def:IC}
A policy $\pi$ is IC if $V_1^{\sfA,\pi,\perp}(\varnothing) \ge V_1^{\sfA,\pi,\rho} (\varnothing)$ for every possible deviation plan $\rho$ of the agent.
\end{definition}

\section{Computing an Optimal Policy}
\label{sc:computing}

We use a dynamic programming approach and compute a near-optimal policy by constructing the {\em inducible value sets}.
The approach is similar to solving an MDP by reasoning about the values of the states. The difference is that, since we are in a two-player setting and need to manage both players' incentives, we use a two-dimensional value, i.e., a value vector, to capture both players' values. We compute the set of all possible value vectors that can be induced by some policy of the principal.

\begin{definition}[Inducible value set]
The inducible value set $\calV_h(\sigma) \subseteq \mathbb{R}^2$ of a sequence $\sigma \in \Sigma_{h-1}$ consists of all vectors $\vv = (v^\sfP, v^\sfA)$ such that $v^\sfP = V_h^{\sfP,\pi,\rho}(\sigma)$ and $v^\sfA = V_h^{\sfA,\pi,\rho}(\sigma)$ 
for some policy $\pi$ and deviation plan $\rho \in \argmax_{\rho'} V_h^{\sfA,\pi,\rho'}(\sigma)$.
\end{definition}

By definition, it is straightforward that once we obtain $\calV_1(\varnothing)$, the principal's optimal value in \eqref{eq:main-prob} can be computed by solving $\max_{(v^\sfP, v^\sfA) \in \calV_1(\varnothing)} v^\sfP$.
A key observation is that the value sets are the same for sequences that end with the same state-action pair. Hence, it suffices to construct one set for each state-action pair, rather than dealing with each of the (exponentially many) possible sequences.
Intuitively, given the state-action pair in the previous time step, the current state and the subsequent process is independent of the earlier history. For ease of description, in what follows, we denote by $O = S \times A$ the set of all possible state-action pairs.

\begin{restatable}{lemma}{lmmsigmao}
\label{lmm:sigma-o}
For all $\sigma,\sigma' \in \Sigma_{h-1}$, it holds that
$\calV_h(\sigma) = \calV_h(\sigma')$ if $o_{h-1} = o'_{h-1}$, where $o_{h-1}, o'_{h-1} \in O$ are the state-action pairs in time step $h-1$, in $\sigma$ and $\sigma'$, respectively.
\end{restatable}

Given the above lemma, we will denote by $\calV_h(o)$ the value set of all sequences ending with $o$. Namely, for all $\sigma \in \Sigma_{h-1}$ in which $(s_{h-1}, \va_{h-1}) = o$, we have $\vv \in \calV_h(o)$ if and only if $\vv \in \calV_h(\sigma)$.
We construct the value sets via a dynamic programming approach next.

\subsection{Computing Inducible Value Sets}
\label{sc:backward-induction}

We will use a convex polytope to approximate each inducible value set. Let $\wh{\calV}_h(o)$ denote the approximation of $\calV_h(o)$ we aim to obtain. Recall that in the last time step all rewards are $0$, so trivially we use $\wh{\calV}_H(o) = \calV_H(s) = \{(0,0)\}$ for all $o \in O$ as the base case.

\paragraph{Dynamic Programming}
Now suppose that we have obtained the polytopes $\wh{\calV}_{h+1}(o')$ for all $o' \in O$. We move to time step $h$ and construct each $\wh{\calV}_h(o)$ based on the $\wh{\calV}_{h+1}(o')$'s.
Central to the approach is the following characterization, which describes an IC condition at time step $h$: 
for every $\vv \in \mathbb{R}^2$, it holds that $\vv \in \calV_h(o)$ if and only if there exist a one-step policy $\varpi: \Omega \to \Delta(A)$ and a set of onward value vectors 
$\left\{\vv'(\bar{\sigma}) \in \mathbb{R}^2: \bar{\sigma} \in \bar{\Sigma} \right\}$ 
that satisfy the following constraints.
\begin{enumerate}[leftmargin=5mm,itemsep=1mm]
\item A value function constraint based on \Cref{eq:V}, which expresses 
$\vv$ via the immediate rewards and onward value vectors $\vv'$ to be induced next, assuming truthful response of the agent:
\begin{align}
\vv 
=&
\sum_{s, \vomega, \va} 
p_{h-1}(s, \vomega \given o) \cdot \hltb{\varpi({\va} \given \vomega)} \cdot 
\Big( \rr_h(s, \va)  + 
\hltb{\vv'(s, \vomega, \omega^\sfA, \va, a^\sfA)}  \Big),  \label{eq:tV-omega-t-cons-v}
\end{align}
The onward value vectors represent the subsequent part of the principal's commitment, which is contingent on the interaction $(s, \vomega, \wtomega^\sfA, \va, \tda^\sfA)$ in time step $h$. They can be viewed as part of the principal's strategy, as if the principal directly selects the future values.
Under the truthful response of the agent, we have $\wtomega^\sfA = \omega^\sfA$ and $\tda^\sfA = a^\sfA$ in \Cref{eq:tV-omega-t-cons-v}.

\item IC constraints, which ensure that the agent's truthful behavior assumed in \Cref{eq:tV-omega-t-cons-v} is indeed incentivized, where we denote by $p_{h-1}(s, \omega^\sfP \given o, \omega^\sfA) \propto p_{h-1}(s, \vomega \given o)$ the conditional probability defined by $p_{h-1}$:
\begin{align}
% \label{eq:tV-omega-t-cons-ic}
& 
\sum_{s, \omega^\sfP, \va}
p_{h-1}(s, \omega^\sfP \given o, \omega^\sfA) \cdot \hltb{\varpi ( {\va} \given \vomega )} \cdot \left( r_h^\sfA(s, \va)  + 
% \sum_{o'} \phi_h(o' \given s, \va) \cdot 
\hltb{{v'}^\sfA (s, \vomega, \omega^\sfA, \va, a^\sfA )}  \right) \ge \nonumber \\
& 
\sum_{a^\sfA}\ 
\max_{\tda^\sfA \in A^\sfA}\
\sum_{s, \omega^\sfP, a^\sfP}\
p_{h-1}(s, \omega^\sfP \given o, \omega^\sfA) \cdot \hltb{\varpi ( \va  \given \omega^\sfP, \wtomega^\sfA )} \cdot \left( r_h^\sfA (s, a^\sfP, \tda^\sfA )  + 
% \phantom{{v'}^\sfA}\right. 
% \nonumber\\
% &\quad\qquad\qquad\qquad
% \left.
% \sum_{o'} \phi_h \left(o' \given s, a^\sfP, \tda^\sfA \right) \cdot 
\hltb{{v'}^\sfA (s, \vomega, \wtomega^\sfA, \va, \tda^\sfA )}  \right) 
\nonumber\\
& \qquad\qquad\qquad\qquad\qquad\qquad\qquad\qquad\qquad\qquad\qquad\qquad\qquad\qquad\qquad 
\text{for all } \omega^\sfA \in \Omega^\sfA
\label{eq:tV-omega-t-cons-ic}
\end{align}
% \vspace{-4mm}
% \end{framed}
% \end{figure*}
Namely, the constraint says, upon observing $\omega^\sfA$, the agent's expected payoff under their truthful response is at least as much as what they could have obtained, had they: 1) reported a different observation $\wtomega^\sfA$, 2) performed a best action $\tda^\sfA$ in response to every possible recommendation $a^\sfA$ of the principal, and 3) responded optimally in the subsequent time steps (whereby the onward values are given by $\vv'$).

\item Onward value constraints, which ensures that the onward values given by $\vv'$ are also inducible:
\begin{align}
\label{eq:tV-omega-t-cons-v-in-V-exact}
\hltb{\vv'(s, \vomega, \wtomega^\sfA, \va, \tda^\sfA)} \in {\calV}_{h+1} (s, a^\sfP, \tda^\sfA )
&& \text{for all } 
(s, \vomega , \wtomega^\sfA, \va , \tda^\sfA) \in \bar{\Sigma}.
\end{align}
\end{enumerate}

The following lemma indicates the correctness of the above characterization.

\begin{restatable}{lemma}{lmminducibleset}
\label{lmm:inducible-set}
$\vv \in \calV_h(o)$ if and only if there exist $\varpi: \Omega \to \Delta(A)$ and $\vv': \bar{\Sigma} \to \mathbb{R}^2$
such that \Cref{eq:tV-omega-t-cons-v,eq:tV-omega-t-cons-ic,eq:tV-omega-t-cons-v-in-V-exact} hold.
\end{restatable}

Therefore, to decide whether $\vv \in {\calV}_h(o)$ amounts to deciding whether the above constraints are satisfied by some $\varpi$ and $\vv'$ (highlighted in blue in the constraints). 
Note that since the inductive hypothesis assumes an approximation $\wh{\calV}_{h+1}(o')$ instead of the exact set $\calV_{h+1}(o')$, we will in fact impose the following {\em approximate} onward value constraint, instead of the exact version in \Cref{eq:tV-omega-t-cons-v-in-V-exact}:
\begin{align}
& \hltb{\vv'(s, \vomega, \wtomega^\sfA, \va, \tda^\sfA)} \in \wh{\calV}_{h+1} (s, a^\sfP, \tda^\sfA )  
&& \text{for all } (s, \vomega , \wtomega^\sfA, \va , \tda^\sfA) \in \bar{\Sigma}.
\label{eq:tV-omega-t-cons-v-in-V}
\end{align}

\paragraph{Linearizing \Cref{eq:tV-omega-t-cons-v,eq:tV-omega-t-cons-ic}}
The constraint satisfiability problem defined above is non-linear due to the quadratic terms and the maximization operator in \Cref{eq:tV-omega-t-cons-v,eq:tV-omega-t-cons-ic}.
Nevertheless, it can be linearized as long as every polytope $\wh{\calV}_{h+1}(o')$, $o' \in O$, is given by the {\em half-space representation}, i.e., by linear constraints in the form $\mathbf{H} \cdot \xx \le \bb$ for some matrix $\mathbf{H}$ and vector $\bb$. 
Specifically, to remove the maximization operator in \Cref{eq:tV-omega-t-cons-ic}, we introduce a set of auxiliary variables 
$y(a^\sfA, \omega^\sfA, \wtomega^\sfA)$ 
to capture the maximum values on the right hand side of \Cref{eq:tV-omega-t-cons-ic}.
We replace the right hand side of \Cref{eq:tV-omega-t-cons-ic} with
$\sum_{ a^\sfA \in A^\sfA } y(a^\sfA, \omega^\sfA, \wtomega^\sfA)$,
and by adding the following constraint we force each $y(a^\sfA, \omega^\sfA, \wtomega^\sfA)$ to be an upper bound of the corresponding maximum value:
\begin{align}
\label{eq:y-ge} 
\hltb{y(a^\sfA, \omega^\sfA, \wtomega^\sfA)} 
&\ge 
\sum_{s, \omega^\sfP, a^\sfP}\
p_{h-1}(s, \omega^\sfP \given o, \omega^\sfA) \cdot \hltb{\varpi ( \va  \given \omega^\sfP, \wtomega^\sfA )} \cdot 
% \nonumber\\
% &\text{\small $
\left( r_h^\sfA (s, a^\sfP, \tda^\sfA ) + 
% \phantom{\sum_{o' }} \right.
% \nonumber\\
% &\qquad\qquad\qquad
% \left.
% \sum_{o'} \phi_h \left(o' \given s, a^\sfP, \tda^\sfA \right) \cdot 
\hltb{{v'}^\sfA(s, \vomega, \wtomega^\sfA, \va, \tda^\sfA)}  \right)
% $ }  
\nonumber \\
&\qquad\qquad\qquad\qquad\qquad\qquad\qquad\qquad\qquad\qquad\qquad\qquad\qquad 
\text{for all } \tda^\sfA \in A^\sfA
\end{align}
To remove the quadratic terms in \Cref{eq:tV-omega-t-cons-v,eq:y-ge}, we use an auxiliary variable $\zz(s, \vomega, \wtomega^\sfA, \va, \tda^\sfA)$ to replace each term 
$\varpi ( \va  \given \vomega ) \cdot \vv'(s, \vomega, \wtomega^\sfA, \va, \tda^\sfA)$ and impose the following constraint on $\zz$:
\begin{equation}
\label{eq:tV-omega-linear-cons-z-Ab}
\mathbf{H} \cdot  \hltb{\zz(s, \vomega, \wtomega^\sfA, \va, \tda^\sfA)} 
\le \hltb{\varpi ( \va  \given \vomega )} \cdot \bb,
\end{equation}
where $\mathbf{H}$ and $\bb$ are taken from the half-space representation of the polytope $\wh{\calV}_{h+1}$, i.e., $\wh{\calV}_{h+1}(s, a^\sfP, \tda^\sfA) = \{\xx : \mathbf{H} \cdot \xx \le \bb \}$.
It is straightforward that, when $\wh{\calV}_{h+1}(s, a^\sfP, \tda^\sfA)$ is nonempty and bounded, \Cref{eq:tV-omega-linear-cons-z-Ab} holds if and only if $\zz(s, \vomega, \wtomega^\sfA, \va, \tda^\sfA) = \varpi ( \va  \given \vomega ) \cdot \xx$ for some $\xx \in \wh{\calV}_{h+1}(s, a^\sfP, \tda^\sfA)$.\footnote{Note that if $\varpi ( \va  \given \vomega ) = 0$, then \Cref{eq:tV-omega-linear-cons-z-Ab} imply that $\zz(s, \vomega, \wtomega^\sfA, \va, \tda^\sfA) = \mathbf{0}$: otherwise, the fact that $\xx' = c \cdot \zz(s, \vomega, \wtomega^\sfA, \va, \tda^\sfA) + \xx$ satisfies $\mathbf{H} \cdot \xx' \le \bb$ for any $c \ge 0$ and $\xx \in \wh{\calV}_{h+1}(s, a^\sfP, \tda^\sfA)$ would prevent $\wh{\calV}_{h+1}(s, a^\sfP, \tda^\sfA)$ from being bounded.}
Hence, \Cref{eq:tV-omega-linear-cons-z-Ab} is the only constraint needed (for each tuple $(s, \vomega, \wtomega^\sfA, \va, \tda^\sfA)$) after we replace the terms with $\zz$. 
This completes the linearization of \Cref{eq:tV-omega-t-cons-v,eq:tV-omega-t-cons-ic}. The complete formulation of the linear constraint satisfiability problem can be found in \Cref{sec:lp}.

\begin{figure}
\centering
\begin{framed}
\raggedright
For $h = H-1,\, \dots,\, 1$,
do the following for all $o \in O$:

\begin{enumerate}
[itemsep=3pt,leftmargin=5mm]
\item 
Plug in \Cref{eq:tV-omega-t-cons-v-in-V} the half-space representation of $\wh{\calV}_{h+1}(o')$, $o' \in O$. Then linearize \Cref{eq:tV-omega-t-cons-v,eq:tV-omega-t-cons-ic}.

\item 
Discretize the space $[0,H]^2$ into a finite point set (see \Cref{lmm:compute-calV} for more detail). 
Check the inducibility of each point $\vv$ in this set by solving the linear constraint satisfiability problem defined by (the linearized version of) \Cref{eq:tV-omega-t-cons-v,eq:tV-omega-t-cons-ic,eq:tV-omega-t-cons-v-in-V}.

\item 
Compute $\wh{\calV}_h(o)$ as the convex hull of the inducible points obtained above, in half-space representation.
\end{enumerate}
\vspace{-2mm}
\end{framed}
\vspace{-2mm}
\caption{Computing approximate value polytopes via dynamic programming.}
\label{fig:dp}
\end{figure}

\paragraph{Constructing $\wh{\calV}_h(o)$}
As a result, we obtain a polytope $\calP$ defined by a set of linear constraints equivalent to \Cref{eq:tV-omega-t-cons-v,eq:tV-omega-t-cons-ic,eq:tV-omega-t-cons-v-in-V}.
The projection of $\calP$ onto the dimensions of $\vv$ is (approximately) $\calV_h(o)$.
To ensure that the projection can be plugged back into \Cref{eq:tV-omega-t-cons-v-in-V} in the next induction step, we need the half-space representation of the projection, too. 
In particular, we want to eliminate the additional variables in the representation so that only $\vv$ remains. (Otherwise, the number of variables may grow exponentially as the induction step increases.)
This can be done approximately in polynomial time given that $\vv$ is two-dimensional. Roughly speaking, we discretize the box $[0, H]^2$ into a finite set of points (recall that rewards in each time step are bounded in $[0,1]$, so $[0, H]^2$ contains $\calV_h(o)$), check the inducibility of each point, and compute the convex hull of the inducible points in half-space representation.
The specific way we discretize the space (see \Cref{fig:slice}) ensures that IC is satisfied {\em exactly} (which can otherwise not be achieved by using standard grid-based discretization). The details can be found in the proof of  \Cref{lmm:compute-calV}. 

\medskip

Repeating the induction procedure till $h=1$, we obtain $\wh{\calV}_1(\varnothing)$ as well as a near-optimal value of the principal by solving the LP $\max_{\vv \in \wh{\calV}_1(\varnothing)} v^\sfP$.
This dynamic programming approach is summarized in \Cref{fig:dp}.

\begin{restatable}{lemma}{thmcomputecalV}
\label{lmm:compute-calV}
For any constant $\epsilon > 0$, it can be computed in time $\poly(|S| {\cdot} |A| {\cdot} |\Omega|, H, 1/\epsilon)$ the half-space representations of a set of polytopes $\wh{\calV}_h(o) \subseteq \calV_h(o)$, $o \in O \cup\{\varnothing\}$ and $h=1,\dots,H$, such that \Cref{eq:tV-omega-t-cons-v,eq:tV-omega-t-cons-ic,eq:tV-omega-t-cons-v-in-V} are satisfiable for every $\vv \in \wh{\calV}_h(o)$ and
$\max_{\vv \in \wh{\calV}_1(\varnothing)} v^\sfP \ge \max_{\vv \in \calV_1(\varnothing)} v^\sfP - \epsilon$.
\end{restatable}

\subsection{Forward Computation of Optimal Policy}
\label{sc:forward-construction}

The above procedure yields the maximum inducible value of the principal but not yet an optimal policy that achieves this value.
We next demonstrate how to compute an optimal policy based on $\wh{\calV}_1(\varnothing)$.
Rather than obtaining an explicit description of a history-dependent policy $\pi$---which would be exponentially large as the policy specifies a distribution for each possible sequence---we present an efficient procedure that computes the distribution $\pi(\cdot \given \sigma; \omega^\sfP, \wtomega^\sfA)$ for any given sequence $(\sigma; \omega^\sfP, \wtomega^\sfA)$.
This means that, when playing the game, the principal can compute an optimal policy on-the-fly based on the realized history.

\begin{figure}
\centering
\begin{framed}
\raggedright
Input: a sequence $(\sigma; \omega^\sfP, \wtomega^\sfA)$, where $\sigma = (s_\ell, \vomega_\ell, \wtomega_\ell^\sfA, \va_\ell, \tda_\ell^\sfA)_{\ell=1}^{h-1}$.

\begin{enumerate}
[itemsep=3pt,leftmargin=5mm]
\item Initialize:
$\vv \leftarrow \argmax_{\vv \in \wh{\calV}_1(\varnothing)} v^\sfP$ and $o \leftarrow \varnothing$.

\item For $\ell = 1, \dots, h - 1$:

\begin{itemize}[itemsep=3pt,leftmargin=5mm]

\item Fix $\vv$ and $o$, and solve \Cref{eq:tV-omega-t-cons-v,eq:tV-omega-t-cons-ic,eq:tV-omega-t-cons-v-in-V}, where we use the polytopes $\wh{\calV}_h(o)$ described in \Cref{lmm:compute-calV}.
Let the solution be $\varpi$ and $\vv'$.

\item Update: 
$\vv \leftarrow \vv'(s_\ell, \vomega_\ell, \wtomega_\ell^\sfA, \va_\ell, \tda_\ell^\sfA)$ and $o \leftarrow (s_\ell, a_\ell^\sfP, a_\ell^\sfA)$. 
\end{itemize}

\item
Output $\pi(\cdot \given \sigma; \omega^\sfP, \wtomega^\sfA) = \varpi(\cdot \given \omega^\sfP, \wtomega^\sfA)$.
\end{enumerate}
\vspace{-2mm}
\end{framed}
\vspace{-2mm}
\caption{Computing a near-optimal policy based on approximations of the value polytopes.}
\label{fig:forward}
\end{figure}

We use a forward computation procedure presented in \Cref{fig:forward}. Starting from time step $1$, the procedure repeatedly computes a one-step policy $\varpi$ and a set of onward vectors, to induce the target value vector $\vv$.
The onward vectors define the target values to be induced in the next time step, contingent on the interaction in the current, which is given by $\sigma$.
Hence, the target vector is updated to one of the onward vectors according $\sigma$ at the end of each iteration.
In other words, in each time step, we expand the target vector into a set of onward vectors, and then select one of them as the next target vector according to the realized interaction given by $\sigma$. 

\medskip

This leads to the following main result of this section.

\begin{restatable}{theorem}{thmopt}
\label{thm:opt}
There exists an $\epsilon$-optimal IC policy $\pi$ such that, for any given sequence $(\sigma; \omega^\sfP, \wtomega^\sfA) \in \Sigma \times \Omega$, the distribution $\pi(\cdot \given \sigma; \omega^\sfP, \wtomega^\sfA)$ can be computed in time $\poly(|S| {\cdot} |A| {\cdot} |\Omega|, H, 1/\epsilon)$.
\end{restatable}

\section{Learning to Commit}
\label{sc:learning}

\newcommand{\whmu}{\widehat{\mu}}
\newcommand{\whphi}{\widehat{\phi}}

We now turn to an episodic online learning setting where the transition model $p: S \times A \rightarrow \Delta(S \times \Omega)$ is not known to the players beforehand. Let there be $T$ episodes. At the beginning of each episode, the principal commits to a new policy based on the outcomes of the previous episodes. Each episode proceeds in $H$ time steps the same way as the model defined in \Cref{sec:prelim}.

We present a learning algorithm that guarantees sublinear regrets for both players under hindsight observability. 
The algorithm is {\em centralized} and relies on the agent behaving truthfully. It does not guarantee exact IC during the course of learning but IC in the limit when the number of episodes approaches infinity. 
Indeed, since the model is unknown to both players, IC in the limit is a more relevant concept as the agent cannot decide how to optimally deviate from their truthful response, either. 
In this case, the sublinear regret the algorithm guarantees for the agent should in many scenarios be sufficient for incentivizing for the agent to participate and follow the centralized learning protocol.

The players' regrets are defined as follows:
\begin{equation*}
\reg^\sfP = \sum_{t=1}^T \left( V^* - V_1^{\sfP,\pi_t,\perp}(\varnothing) \right)
\qquad\text{and}\qquad
\reg^\sfA = \sum_{t=1}^T \left( \max_{\rho} V_1^{\sfA,\pi_t, \rho}(\varnothing) - V_1^{\sfA,\pi_t,\perp}(\varnothing) \right),
\end{equation*}
where $V^*$ is the optimal value of \Cref{eq:main-prob} and $\pi_t$ denotes the policy the principal commits to in the $t$-th episode. 
In words, the principal's regret $\reg^\sfP$ is defined with respect to the optimal policy under the true model. The agent's regret $\reg^\sfA$ is defined with respect to his optimal response to each $\pi_t$, which is a dynamic regret as the benchmark changes across the episodes.

\subsection{Learning Algorithm}

\paragraph{Reward-free Exploration}
Our learning algorithm is based on {\em reward-free exploration}, which is an RL paradigm where learning happens before a reward function is provided \cite{jin2020reward}. It has been shown in a series of works that efficient learning is possible under this paradigm \cite{jin2020reward,kaufmann2021adaptive,menard2021fast}.
In particular, we will use the sample complexity bound in \Cref{lmm:empirical-model-bound}.
At a high level, our algorithm proceeds by first conducting reward-free exploration to learn a sufficiently accurate estimate of the true model. Based on the estimate we then solve a relaxed version of the policy optimization problem \Cref{eq:main-prob} to obtain a policy. Using this policy in the remaining episodes guarantees sublinear regret for both players.

\begin{lemma}[{\citep[Lemma~3.6 restated]{jin2020reward}}]
\label{lmm:empirical-model-bound}
Consider an (single-player) MDP $(S, A, p)$ (without any reward function specified) with horizon length $H$. 
There exists an algorithm which learns a model $\whp$ after $\wt{\calO} \left( \frac{H^5 \abs{S}^2 \cdot \abs{A}} {\delta^2} \right)$ episodes of exploration, such that with probability at least $1- q$, for any reward function $r$ and policy $\pi$, it holds that 
\[
\abs{V_1^{r,\pi}(s) - \whV_1^{r,\pi}(s)} \le \delta/2
\]
for all states $s$, where $V_1^{r,\pi}$ and $\whV_1^{r,\pi}$ denote the value functions under reward function $r$ and models $p$ and $\whp$, respectively.\footnote{The notation $\wt{\calO}$ omits logarithmic factors. In the original statement of \citet{jin2020reward}, $\pi$ is non-stationary (time-dependent) but {independent} of the history. However, the proof of the lemma also applies to history-dependent policies. The dependence on $H$ in the sample complexity can be further improved with better reward-free exploration algorithms \cite{kaufmann2021adaptive,menard2021fast}, but this is not a focus of ours.}
\end{lemma}

With the above result, we can learn a model $\whp$ for our purpose. 
In what follows, we let $\wh{V}_h^{\sfP,\pi,\rho}$ and $\wh{V}_h^{\sfA,\pi,\rho}$ denote the players' value functions in model $\whp$ (i.e., by replacing $p$ in \Cref{eq:V} with $\whp$). 
\Cref{lmm:Pi-hat-bounds} then translates \Cref{lmm:empirical-model-bound} to our setting.
Note that under hindsight observability the process facing the principal and the agent jointly during the learning process is effectively an MDP, where the effective state space is $O\times \Omega$.
An effective state, say $\theta = (s, \va, \vomega)$, consists of the state-action pair $(s,\va)$ in the previous step and the observations $\vomega$ in the current. When a joint action $\va'$ is performed, $\theta$ transitions to $\theta' = (s', \va', \vomega')$ with probability $p_{h-1}(s', \vomega' \given s, \va)$.

\begin{lemma}
\label{lmm:Pi-hat-bounds}
A model $\whp$ can be learned after $\wt{\calO} \left( \frac{H^5 \abs{S}^2 \abs{A}^3 \abs{\Omega}^2} {\delta^2} \right)$ episodes of exploration, such that $\abs{ V_1^{\sfA,\pi,\rho}(\varnothing) - \wh{V}_1^{\sfA,\pi,\rho}(\varnothing)} \le \delta/2$ and $\abs{ V_1^{\sfP,\pi,\rho}(\varnothing) - \wh{V}_1^{\sfP,\pi,\rho}(\varnothing)} \le \delta/2$ with probability at least $1- q$ for any policy $\pi$ and deviation plan $\rho$.
\end{lemma}

Therefore, the value functions change smoothly as the learned model $\whp$ approaches $p$.
However, this smoothness is insufficient for deriving a sublinear bound on the principal's regret because of the agent's incentive constraints in our problem.
Roughly speaking, the set of IC policies does not change smoothly with $\whp$, even though the value functions do. Hence, even an infinitesimal difference between $\whp$ and $p$ may lead to a jump between the IC policy sets under these two models and, in turn, a gap between the values of the optimal policies.

\paragraph{Approximate IC Relaxation}
To deal with this issue, we relax the incentive constraints, allowing small violations to the constraints. Such violations are inevitable if we aim to achieve a near-optimal value under the true model $p$ but only know an estimate $\whp$ of the true model. On the positive side, given the sublinear regret guarantee for the agent, the violation diminishes with the number of episodes.
We define $\delta$-IC policies below.

\begin{definition}[$\delta$-IC policy]
A policy $\pi$ is $\delta$-IC (w.r.t. model $\whp$) if $\wh{V}_1^{\sfA,\pi,\perp}(\varnothing) \ge \wh{V}_1^{\sfA,\pi,\rho} (\varnothing) - \delta$ for every possible deviation plan $\rho$ of the agent.
A $\delta$-IC policy is said to be $\epsilon$-optimal if $\wh{V}_1^{\sfP,\pi,\perp}(\varnothing) \ge V^* - \delta$, where $V^*$ is the optimal value of \Cref{eq:main-prob} (under $p$).
\end{definition}

That is, in response to a $\delta$-IC policy, the agent can improve his overall expected payoff by no more than $\delta$ if he deviates from the truthful response. 
We assume that the agent will not deviate for such a small benefit, and we evaluate the value of a $\delta$-IC policy based on the agent's truthful response. This is how the $\epsilon$-optimality is defined above, where we compare against the optimal value $V^*$ in \Cref{eq:main-prob}, which is obtained under a more stringent setting without any relaxation of the agent's incentive.
In other words, we relax the feasible space and compare the solution obtained in this relaxed space with the optimum over the smaller original feasible space. Such relaxations are common in the optimization literature, and they are crucial for resolving the non-smooth issue.

Let $\wh{\Pi}_\delta$ and $\Pi_\delta$ denote the set of $\delta$-IC policies under $\whp$ and $p$, respectively.
The relaxation immediately results in $\wh{\Pi}_{\delta} \supseteq \Pi_0$ for the model $\whp$ stated in \Cref{lmm:Pi-hat-bounds}.
As a result, optimizing over $\wh{\Pi}_{\delta}$ ensures that the optimal value yielded is as much (up to a small error) as the optimal value $V^*$ over $\Pi_0$.
Meanwhile, the value loss introduced by this relaxation for the agent is also small (bounded by $\delta$).

\medskip

With the above results, our learning algorithm proceeds as follows.

\begin{framed}
\vspace{-5mm}
\begin{enumerate}
[itemsep=3pt,leftmargin=3mm]
\item Run reward-free exploration to obtain a model $\whp$ as stated in \Cref{lmm:Pi-hat-bounds}.

\item 
Compute a $\delta$-optimal $\delta$-IC policy in $\whp$ and use it in the remaining episodes.
\end{enumerate}
\vspace{-3mm}
\end{framed}

The near-optimal policy in Step~2 can be computed efficiently according to \Cref{lmm:opt-Pi-hat-eps}, via an approach similar to the one in \Cref{sc:backward-induction}.
This gives an efficient algorithm with sublinear regrets for both players. We present \Cref{thm:learning}.

\begin{restatable}{lemma}{lmmoptpihateps}
\label{lmm:opt-Pi-hat-eps}
There exists an $\epsilon$-optimal $\delta$-IC policy $\pi$ such that, for any given sequence $(\sigma; \omega^\sfP, \wtomega^\sfA) \in \Sigma \times \Omega$, the distribution $\pi(\cdot \given \sigma; \omega^\sfP, \wtomega^\sfA)$ can be computed in time $\poly(|S| {\cdot} |A| {\cdot} |\Omega|, H, 1/\epsilon, \log(1/\delta))$.
\end{restatable}

\begin{restatable}{theorem}{thmlearning}
\label{thm:learning}
There exists an algorithm that guarantees regret $\wt{\calO}(\zeta^{1/3} T^{2/3})$ for both players with probability $1-q$, where $\zeta = H^5 \abs{S}^2 \abs{A}^3 \abs{\Omega}^2$. The computation involved in implementing the algorithm takes time $\poly(|S|{\cdot}|A|{\cdot}|\Omega|, H, T)$.
\end{restatable}

\section{Conclusion}
\label{sec:conclusion}

We studied a stochastic principal-agent framework and presented efficient computation and learning algorithms. 
Our model can be further extended to the setting with $n$ agents. The algorithms we presented remain efficient for any constant $n$ if approximate IC solutions are considered. Computing optimal exact IC policies for $n$ agents remain an interesting open question, as our discretization method, which operates by slicing the space, does not generalize to $n$ agents.
When $n$ is not a constant, representing games in normal-form requires space exponential in $n$, so more succinct representations are typically considered. However, it is known that in succinctly represented games even to compute an optimal correlated equilibrium in one-shot games may be NP-hard \cite{papadimitriou2005computing}.)
Our results indicate how a policy designer might interact with agents optimally. In particular implementations, the designer's incentives may not be aligned with societal benefits. In these cases, a careful analysis of the incentives and their moral legitimacy must be considered. Besides this, since the paper is theory focused, we do not feel any other potential impacts must be specifically highlighted here.

\bibliographystyle{plainnat}
% \bibliography{ref}

%%%%%%%%%%%%%%%%%%%%%%%%%%%%%%%%%%%%%%%%%%%%%%%%%%%%%%%%%%%%%%%%%%%%%%%%%%%%%%%
%%%%%%%%%%%%%%%%%%%%%%%%%%%%%%%%%%%%%%%%%%%%%%%%%%%%%%%%%%%%%%%%%%%%%%%%%%%%%%%
% APPENDIX
%%%%%%%%%%%%%%%%%%%%%%%%%%%%%%%%%%%%%%%%%%%%%%%%%%%%%%%%%%%%%%%%%%%%%%%%%%%%%%%
%%%%%%%%%%%%%%%%%%%%%%%%%%%%%%%%%%%%%%%%%%%%%%%%%%%%%%%%%%%%%%%%%%%%%%%%%%%%%%%
\clearpage
\appendix
% \onecolumn

\renewcommand\thefigure{\thesection.\arabic{figure}}

\section{Omitted Proofs}

\subsection{Omitted Proofs in \Cref{sc:computing}}

For simplicity, we write $\vec{V}_h^{\pi,\rho}(\sigma) = \left( V_h^{\sfP,\pi,\rho} (\sigma),\, V_h^{\sfA,\pi,\rho} (\sigma) \right)$ in the following proofs. 

\lmmsigmao*

\begin{proof}
Consider an arbitrary $\vv \in \calV_h(\sigma)$.
By definition, this means that there exists a policy $\pi$ and deviation plan $\rho$ such that $\vv = \vec{V}_h^{\pi,\rho}(\sigma)$ and $\rho = \argmax_{\rho'} V_h^{\sfA,\pi,\rho'}(\sigma)$.
Consider the following policy 
$\pi'$ such that:
$\pi'(\varsigma') = \pi(\varsigma)$ for all sequences $\varsigma',\varsigma \in \Sigma$ which contain $\sigma'$ and $\sigma$, respectively, as subsequences at time steps $1,\dots,h-1$. 
It follows by \Cref{eq:V} that when $o_{h-1} = o'_{h-1}$, we have $\vec{V}_h^{\pi',\rho}(\sigma') = \vec{V}_h^{\pi,\rho}(\sigma)$ and $\rho = \argmax_{\rho'} V_h^{\sfA,\pi',\rho'}$. 
Hence, $\vv \in \calV_h(\sigma')$.
Since the choice of $\vv$ is arbitrary, we get that $\calV_h(\sigma) \subseteq \calV_h(\sigma')$.
By symmetry, it follows that $\calV_h(\sigma) = \calV_h(\sigma')$.
\end{proof}

\lmminducibleset*

\begin{proof}
First, consider the ``only if'' direction of the statement. Suppose that $\vv \in \calV_h(o)$.
By definition, we have $\vv = \vec{V}_h^{\pi,\rho}(\sigma)$ for some $\pi$ and $\rho \in \argmax_{\rho'} \vec{V}_h^{\sfA,\pi,\rho'}(\sigma)$, for all $\sigma \in \Sigma_{h-1}$ that ends with $o$.  
According to a standard revelation principle argument, we can assume w.l.o.g. that $\rho$ is IC in step $h$.
Hence, by \Cref{eq:V}, we have
\begin{align}
\vv 
=&
\sum_{s, \vomega, \va} 
p_{h-1}(s, \vomega \given o) \cdot \pi(\va \given \sigma;\vomega) \cdot 
\Big( \rr_h(s, \va)  + 
\vec{V}_{h+1}^{\pi,\rho}(\sigma; s, \vomega, \omega^\sfA, \va, a^\sfA) \Big).
\end{align}
Letting $\varpi(\va \given \vomega) = \pi(\va \given \sigma; \vomega)$ for every $\vomega \in \Omega$, and $\vv'(\bar{\sigma}) = \vec{V}_{h+1}^{\pi,\rho}(\sigma; \bar{\sigma})$ for every $\bar{\sigma} \in \bar{\Sigma}$, we establish \Cref{eq:tV-omega-t-cons-v}.
Since $\rho$ is IC in step $h$, \Cref{eq:tV-omega-t-cons-ic} also follows immediately: the agent cannot benefit from any possible deviation. 
Finally, by definition, we have $\vv'(\bar{\sigma}) = \vec{V}_{h+1}^{\pi,\rho}(\sigma; \bar{\sigma}) \in \calV_{h+1}(o')$ for every $\bar{\sigma} \in \bar{\Sigma}$ that contains $o'$, so \Cref{eq:tV-omega-t-cons-v-in-V-exact} holds.

Now consider the ``if'' direction. Suppose that \Cref{eq:tV-omega-t-cons-v,eq:tV-omega-t-cons-ic,eq:tV-omega-t-cons-v-in-V-exact} hold for some $\varpi$ and $\vv'$. 
Pick arbitrary $\sigma \in \Sigma_{h-1}$ that ends with $o$. Consider a policy $\pi$ such that:
$\pi(\va \given \sigma; \vomega) = \varpi(\va \given \vomega)$ for all $\vomega \in \Omega$, and 
$\pi(\va \given \sigma; \bar{\sigma};\vomega) = \pi'(\va \given \sigma; \bar{\sigma};\vomega)$ for all $\bar{\sigma} \in \bar{\Sigma}$ and $\vomega \in \Omega$, where $\pi'$ is an arbitrary policy that induces $\vv'(\bar{\sigma})$ for every $\bar{\sigma}$ (which exists given \Cref{eq:tV-omega-t-cons-v-in-V-exact}).
Namely, $\pi$ is the same as $\varpi$ in step $h$ and switches to $\pi'$ in the  subsequent steps.
Given \Cref{eq:tV-omega-t-cons-ic}, the agent cannot benefit from any deviation at step $h$, so \Cref{eq:tV-omega-t-cons-v} gives the players' values for $\pi$ and an optimal deviation plan of the agent. Hence, $\vv \in \calV_h(\sigma) = \calV(o)$.
\end{proof}

\thmcomputecalV*

\begin{proof}
Throughout the proof, we say that the polytope $\wh{\calV}_h(o)$ is an {\em $\varepsilon$-approximation of $\calV_h(o)$} if and only if:
\begin{itemize}
\item $\wh{\calV}_h(o) \subseteq \calV_h(o)$, and 
\item for every $\vv \in \calV_h(o)$, there exists $\vv' \in \wh{\calV}_h(o)$ such that ${v'}^\sfP \ge v^\sfP - \varepsilon$ and ${v'}^\sfA = v^\sfA$.
\end{itemize}
We will show that an $\epsilon$-approximation $\wh{\calV}_1(\varnothing)$ of $\calV_1(\varnothing)$ can be computed efficiently, so that $\max_{\vv \in \wh{\calV}_1(\varnothing)} v^\sfP \ge \max_{\vv \in \calV_1(\varnothing)} v^\sfP - \epsilon$ follows readily.\footnote{In the definition of $\epsilon$-approximation, we require additionally that the projections of $\wh{\calV}_h(o)$ and $\calV_h(o)$ onto the dimension of $v^\sfA$ are the same (i.e., $v'^\sfA = v^\sfA$), so that the approximation compromises only on the principal's value. This is crucial for ensuring exact IC and smooth changes of the approximation throughout the induction process we present below.}
Meanwhile, we also show that the polytopes we compute ensures that \Cref{eq:tV-omega-t-cons-v,eq:tV-omega-t-cons-ic,eq:tV-omega-t-cons-v-in-V} are satisfiable for every $\vv \in \wh{\calV}_h(o)$.

We now prove by induction. The key is the following induction step. Suppose that the following conditions hold for all $o \in O$: 
\begin{itemize}
\item[1.]
% {\bf Property 1.}
$\wh{\calV}_{h+1}(o)$ is defined by $\calO(H/\delta)$ many linear constraints.

\item[2.] 
% {\bf Property 2.}
$\wh{\calV}_{h+1}(o)$ is an $\varepsilon$-approximation of $\calV_{h+1}(o)$.
\end{itemize}
We show that, given the above conditions, for every $o \in O$ we can compute in time polynomial in $1/\delta$ a polytope $\wh{\calV}_h(o)$ (in half-space representation) that satisfies the above conditions (for $h$), with an approximation factor $\varepsilon' = \varepsilon + \delta$ in the second condition.
Once this holds, picking $\delta = \epsilon/ H$ then gives, by induction, that $\wh{\calV}_1(\varnothing)$ is an $\epsilon$-approximation of $\calV_1(\varnothing)$ (where $\epsilon$ is the target constant in the statement of the lemma). Note that as a based case, $\{(0,0)\}$ is readily a $0$-approximation of $\calV_H(o)$ and can be defined by three linear constraints.

We proceed as follows.
For every $o \in O$, let $\overline{\calV}_h(o)$ denote the set of vectors $\vv$ satisfying \Cref{eq:tV-omega-t-cons-v,eq:tV-omega-t-cons-ic,eq:tV-omega-t-cons-v-in-V}.\footnote{Note that $\overline{\calV}_h(o)$ is different from $\calV_h(o)$: the latter, according to \Cref{lmm:inducible-set}, is defined by \Cref{eq:tV-omega-t-cons-v,eq:tV-omega-t-cons-ic,eq:tV-omega-t-cons-v-in-V-exact}, where \Cref{eq:tV-omega-t-cons-v-in-V-exact} uses the exact value sets $\calV_{h+1}(o')$, unlike the approximate ones $\wh{\calV}_{h+1}(o')$ in \Cref{eq:tV-omega-t-cons-v-in-V}.}
We follow the algorithm presented in \Cref{fig:dp} and discretize $[0,H]^2$ to construct $\wh{\calV}_h(o)$.
Specifically, we slice the space
along the dimension of the principal's value. 
We compute the intersection points of the slice lines and (the boundary of) $\overline{\calV}_h(o)$, and construct $\wh{\calV}_h(o)$ as the convex hull of the intersection points to approximate $\overline{\calV}_h(o)$.
Specifically, let 
$W = \left\{0,\ \delta,\ 2\delta,\ \dots,\ H - \delta,\ H \right\}$ contain the principal's values on the slice lines we use, and let $\calW$ be the set consisting of the following points.
\begin{itemize}
    \item 
First, for each $w \in W$, the two intersection points of the slice line at $w$ and $\overline{\calV}_h(o)$:
\begin{align*}
\check{\vv}_w \in \argmin\nolimits_{\vv \in \overline{\calV}_h(o): v^\sfP = w} v^\sfA 
\quad\text{and}\quad
\hat{\vv}_w \in \argmax\nolimits_{\vv \in \overline{\calV}_h(o): v^\sfP = w} v^\sfA.
\end{align*}

\item 
Moreover, two vertices of $\overline{\calV}_h(o)$ with the minimum and maximum values for the agent:
\begin{align*}
\check{\vv}_* \in \argmin\nolimits_{\vv \in \overline{\calV}_h(o)} v^\sfA 
\quad\text{and}\quad
\hat{\vv}_* \in \argmax\nolimits_{\vv \in \overline{\calV}_h(o)} v^\sfA.
\end{align*}
If there are multiple maximum (or minimum) vertices, we pick an arbitrary one.
\end{itemize}
An illustration is given in \Cref{fig:slice}.

\begin{figure}
\centering
\begin{tikzpicture}[scale=0.8]
\tikzset{myslice/.style={-,densely dashed,color=gray}}
\tikzset{mydot/.style={circle,fill,inner sep=0.9pt}}
\tikzset{mycross/.style={cross out, draw=black, thick, minimum size=2*(#1-\pgflinewidth), inner sep=0pt, outer sep=0pt}, mycross/.default={2.5pt}}
    % grid and axes   
    % \draw[->,thin] (-4.6,5.0) -- (-3.5,5.0) node[right] {\small$v^\sfA$};
    % \draw[->,thin] (-4.5,5.1) -- (-4.5,4) node[below] {\small$v^\sfP$};
    \draw[->,thin] (-4.0,5.2) -- (-2.5,5.2) node[above,midway] {\small$v^\sfA$};
    \draw[->,thin] (-4.5,4.7) -- (-4.5,3.2) node[left,midway] {\small$v^\sfP$};
    
    % slice
    \draw[myslice, name path=x1] (-3.6,1) -- (1.8,1) node {};
    \draw[myslice, name path=x2] (-3.6,2) -- (1.8,2) node {};
    \draw[myslice, name path=x3] (-3.6,3) -- (1.8,3) node {};
    \draw[myslice, name path=x4] (-3.6,4) -- (1.8,4) node {};
    \draw[dotted,thick,gray] (-3,0.5) -- (-3,4.5) node {};
    \draw[dotted,thick,gray] (0.9,0.5) -- (0.9,4.5) node {};

    \coordinate (a) at (-3,2.6)  ;
    \coordinate (b) at (-2.3,3.7);
    \coordinate (c) at (-1.0,4.3);
    \coordinate (d) at (0.4,3.6) ;
    \coordinate (e) at (0.9,2.4) ;
    \coordinate (f) at (0.7,1.4) ;
    \coordinate (g) at (-0.8,0.5);
    \coordinate (h) at (-1.8,0.7);
    \coordinate (i) at (-2.8,1.4);

    \node[blue!40!gray,fill=white] at (-0.8,0.2) {\small$\overline{\calV}_{h}(o)$};
    \draw (2.1,2) edge[|<->|] node[fill=white]{\small$\delta$} (2.1,1);

    \node[left] at (a) {\small$\check{\vv}_*$};
    \node[right] at (e) {\small$\hat{\vv}_*$};

    % edges  
    \draw[name path=ab,draw=none] (a) -- (b) node {};
    \draw[name path=bc,draw=none] (b) -- (c) node {};
    \draw[name path=cd,draw=none] (c) -- (d) node {};
    \draw[name path=de,draw=none] (d) -- (e) node {};
    \draw[name path=ef,draw=none] (e) -- (f) node {};
    \draw[name path=fg,draw=none] (f) -- (g) node {};
    \draw[name path=gh,draw=none] (g) -- (h) node {};
    \draw[name path=hi,draw=none] (h) -- (i) node {};
    \draw[name path=ia,draw=none] (i) -- (a) node {};

    % calculate intersection points
    \coordinate (p0) at (a);
    \coordinate (p5) at (e);
    \path[name intersections={of=ab and x3,by={p1}}];
    \path[name intersections={of=bc and x4,by={p2}}];
    \path[name intersections={of=cd and x4,by={p3}}];
    \path[name intersections={of=de and x3,by={p4}}];
    \path[name intersections={of=ef and x2,by={p6}}];
    \path[name intersections={of=fg and x1,by={p7}}];
    \path[name intersections={of=hi and x1,by={p8}}];
    \path[name intersections={of=ia and x2,by={p9}}];
    
    \node[above left] at (p2) {\small$\check{\vv}_{\delta}$};
    \node[above right] at (p3) {\small$\hat{\vv}_{\delta}$};
    \node[below right,xshift=-2mm] at (p7) {\small$\hat{\vv}_{4\delta}$};
    \node[below left,xshift=3mm] at (p8) {\small$\check{\vv}_{4\delta}$};
    
    % draw the polygons    
    \filldraw[thick,fill=blue!3!white,draw=blue!50!white,fill opacity=1] 
    (a.center) -- (b.center) -- (c.center) -- (d.center)  -- (e.center) -- (f.center) -- (g.center) -- (h.center) -- (i.center) -- cycle;
    
    \filldraw[thick,draw=red!50!white,fill=red!10!white,fill opacity=1] 
    (p0.center) -- (p1.center) -- (p2.center) -- (p3.center) -- (p4.center) -- (p5.center) -- (p6.center) -- (p7.center) -- (p8.center) -- (p9.center) -- cycle;

    \node[mydot] at (p0) {};
    \node[mydot] at (p1) {};
    \node[mydot] at (p2) {};
    \node[mydot] at (p3) {};
    \node[mydot] at (p4) {};
    \node[mydot] at (p5) {};
    \node[mydot] at (p6) {};
    \node[mydot] at (p7) {};
    \node[mydot] at (p8) {};
    \node[mydot] at (p9) {};

    % labels
    \node[red] at (-0.9,2.5) {\small$\wh{\calV}_{h}(o)$};

    % \node[mycross,rotate=45] at (-2,1.5) {};
    % \node[mycross,rotate=0] at ($(p9)!0.5!(a)$) {};
\end{tikzpicture}
\vspace{5mm}
\caption{Constructing $\wh{\calV}_{h}(o)$ as a $\delta$-approximation of $\overline{\calV}_{h}(o)$. The black points constitute $\calW$ (labels of $\check{\vv}_{2\delta}$, $\check{\vv}_{3\delta}$, $\hat{\vv}_{2\delta}$, and $\hat{\vv}_{3\delta}$ are omitted).}
\label{fig:slice}
\end{figure}

It shall be clear that the choice of these points ensures that we can approximate any inducible value vector with at most $\delta$ compromise on the principal's value and no compromise on the agent's.
(In particular, the inclusion of $\check{\vv}_*$ and $\hat{\vv}_*$ ensures that we do not miss the agent's extreme values that may not be attained at any of the slice lines.)
All the points can be computed efficiently by solving LPs that minimizes (or maximizes) $v^\sfA$ (where we also treat $\vv$ as variables in addition to the other variables), subject to the linearized version of \Cref{eq:tV-omega-t-cons-v,eq:tV-omega-t-cons-ic,eq:tV-omega-t-cons-v-in-V}, 
% eq:tV-omega-t-cons-pi
and additionally $v^\sfP = w$ when we compute $\check{\vv}_w$ or $\hat{\vv}_w$. 
The hypothesis that $\wh{\calV}_{h+1}(o)$ is defined by $\calO(H/\delta)$ linear constraints ensures that all the LPs are polynomial sized and hence can be solved efficiently.

We then compute $\wh{\calV}_h(o)$ as the convex hull of $\calW$. Given that the space is two-dimensional, this can be done efficiently via standard algorithms in computational geometry (e.g., Chan's algorithm \cite{chan1996optimal}). 
This way, the first condition in the inductive hypothesis holds for $\wh{\calV}_h(o)$ because $\wh{\calV}_h(o)$ has at most $\calO(H/\delta)$ vertices while it is in $\mathbb{R}^2$.
Meanwhile, $\wh{\calV}_h(o)$ is an $\delta$-approximation of $\overline{\calV}_h(o)$ according to the following arguments.

\begin{claim}
\label{clm:delta-approx}
$\wh{\calV}_h(o)$ is an $\delta$-approximation of $\overline{\calV}_h(o)$.
\end{claim}

\begin{proof}[Proof of \Cref{clm:delta-approx}]
First, since $\calW \subseteq \overline{\calV}_h(o)$ by construction,
$\wh{\calV}_h(o) \subseteq \overline{\calV}_h(o) \subseteq \calV_h(o)$ holds readily.
It remains to show that for any $\vv \in \overline{\calV}_h(o)$ there exists $\xx \in \wh{\calV}_h(o)$ such that $x^\sfA = v^\sfA$ and $x^\sfP \ge v^\sfP - \delta$.

Let $\calB = \{ \vv' \in \mathbb{R}^2 : i \delta \le v'^\sfP \le (i+1) \delta \}$ be the band between two slice lines that contains $\vv$. 
Consider the relation between $v^\sfA$ and the agent's minimum and maximum values attained at $\calW \cap \calB$.
There can be the following possibilities.
\begin{itemize}
\item {Case 1.} 
$v^\sfA$ lies in between the minimum and maximum values, i.e.,
\[
\min_{\vv' \in \calW \cap \calB} v'^\sfA \le v^\sfA \le \max_{\vv' \in \calW \cap \calB} v'^\sfA.
\]
This means that there must be a point $\xx \in \mathrm{ConvexHull}(\calW \cap \calB)$ such that $x^\sfA = v^\sfA$. We have $\xx \in \mathrm{ConvexHull}(\calW \cap \calB) \subseteq \calB$.
So both $\vv$ and $\xx$ are inside $\calB$. According to the definition of $\calB$, this means $x^\sfP \ge v^\sfP - \delta$, as desired.

\item {Case 2.} 
$v^\sfA < \min_{\vv' \in \calW \cap \calB} v'^\sfA$.
In this case, it must be that $\check{\vv}_* \notin \calB$ (otherwise, $\min_{\vv' \in \calW \cap \calB} v'^\sfA = \check{v}_*^\sfA \le v^\sfA$). 
Now that $\vv \in \calB$, the line segment between $\vv$ and $\check{\vv}_*$ must intersect with the boundary of $\calB$ (i.e., one of the slice lines) at some point $\yy$. We have $y^\sfA \le v^\sfA$ (because $\check{v}_*^\sfA \le v^\sfA$ by definition) and $\yy \in \overline{\calV}_h(o)$ (because $\vv, \check{\vv}_* \in \overline{\calV}_h(o)$). 
Pick $\check{\vv}_w$ where $w = y^\sfP$. By definition $\check{v}_w^\sfA \le y^\sfA$. It follows that
\[
\check{v}_w^\sfA \le y^\sfA \le v^\sfA < \min_{\vv' \in \calW \cap \calB} v'^\sfA.
\]
This is a contradiction because we have $\check{\vv}_w \in \calW \cap \calB$ as $\yy$ is on the boundary of $\calB$.

\item {Case 3.}  
$v^\sfA > \max_{\vv' \in \calW \cap \calB} v'^\sfA$. An argument similar to that for Case~2 implies that this case is not possible, either.
\end{itemize}
Hence, only Case~1 is possible, where a desired point $\xx$ exists. The claim then follows.
\end{proof}

The fact that $\wh{\calV}_h(o) \subseteq \overline{\calV}_h(o)$ also implies that \Cref{eq:tV-omega-t-cons-v,eq:tV-omega-t-cons-ic,eq:tV-omega-t-cons-v-in-V} are satisfiable for every $\vv \in \wh{\calV}_h(o)$, as they are for every $\vv \in \overline{\calV}_h(o)$.
We next confirm that $\wh{\calV}_h(o)$ is eventually an $(\varepsilon + \delta)$-approximation of $\calV_h(o)$.
Indeed, now \Cref{clm:delta-approx} indicates that $\wh{\calV}_h(o)$ is an $\delta$-approximation of $\overline{\calV}_h(o)$, so $\wh{\calV}_h(o)$ is an $(\varepsilon + \delta)$-approximation of $\calV_h(o)$ as long as $\overline{\calV}_h(o)$ is an $\varepsilon$-approximation of $\calV_h(o)$.

To see that $\overline{\calV}_h(o)$ is an $\varepsilon$-approximation, consider an arbitrary $\vv \in \calV_h(o)$.
By \Cref{lmm:inducible-set}, $\vv$ can be induced by some $\varpi$ and $\vv'$ satisfying
\Cref{eq:tV-omega-t-cons-v,eq:tV-omega-t-cons-ic,eq:tV-omega-t-cons-v-in-V-exact}.
By assumption, every $\wh{\calV}_{h+1}(o')$ is an $\varepsilon$-approximation of $\calV_{h+1}(o')$, so 
for every onward vector $\vv'(\bar{\sigma}) \in \calV_{h+1}(o')$, there exists a vector $\tilde{\vv}'(\bar{\sigma}) \in \wh{\calV}_{h+1}(o')$ such that $\tilde{v}'^\sfP(\bar{\sigma}) \ge {v'}^\sfP(\bar{\sigma}) - \varepsilon$ and $\tilde{v}'^\sfA(\bar{\sigma}) = {v'}^\sfA(\bar{\sigma})$.
Using $\tilde{\vv}'$ instead of $\vv'$, the same policy $\varpi$ then induces a vector $\tilde{\vv} \in \wh{\calV}_h(o)$ to approximate $\vv$.
Indeed, the agent's values are exactly the same under $\tilde{\vv}'$ and $\vv'$, so the same response of the agent can be incentivized. This is why we require the approximation to not compromise on the agent's value.
Moreover, according to \Cref{eq:tV-omega-t-cons-v}, the overall difference between $\tilde{v}^\sfP$ and $v^\sfP$ is at most $\varepsilon$ because it holds for the coefficients that $\sum_{s, \vomega, \va} p_{h-1}(s, \vomega \given o) \cdot \varpi(\va \given \vomega ) = 1$.
As a result, $\tilde{v}^\sfP \ge v^\sfP - \varepsilon$ and $\overline{\calV}_h(o)$ is an $\varepsilon$-approximation of $\calV_h(o)$. 

Hence, the inductive hypothesis holds for $h$.
By induction, $\wh{\calV}_1(\varnothing)$ is an $\delta H$-approximation of $\calV_1(\varnothing)$. Since $\delta H = \epsilon$, we get that $\max_{\vv \in \wh{\calV}_1(\varnothing)} v^\sfP \ge \max_{\vv \in \calV_1(\varnothing)} v^\sfP - \epsilon$.
\end{proof}

\thmopt*

\begin{proof}
Consider the algorithm presented in \Cref{fig:forward}.
The outputs of the algorithm over all possible input sequences $(\sigma; \omega^\sfP, \wtomega^\sfA) \in \Sigma \times \Omega$ specify a policy $\pi$.
The polynomial running time of the algorithm for computing each $\pi(\cdot \given \sigma; \omega^\sfP, \wtomega^\sfA)$ follows by noting that it runs by solving at most $H$ linear constraint satisfiability problems.

It remains to argue that $\pi$ is IC and $\epsilon$-optimal.
Indeed, by \Cref{lmm:compute-calV} and an inductive argument, $\pi$ is IC at each time step $h$ and induces the corresponding values encoded in $\vv'$ as the expected onward values.
The $\epsilon$-optimality of $\pi$ follows given the condition $\max_{\vv \in \wh{\calV}_1(\varnothing)} v^\sfP \ge \max_{\vv \in \calV_1(\varnothing)} v^\sfP - \epsilon$ stated in \Cref{lmm:compute-calV} (and the choice of the initial $\vv$ in \Cref{fig:forward}). 
\end{proof}

\subsection{Omitted Proofs in \Cref{sc:learning}}

\lmmoptpihateps*

\begin{proof}
The proof is similar to the approach in \Cref{sc:backward-induction}, which computes a near-optimal and $0$-IC policy.
We describe the differences below.

Instead of maintaining two-dimensional sets of inducible values, we split the dimension of the agent's value into two dimensions $v^\sfA$ and $v_*^\sfA$, which represent the agent's values under his truthful response (i.e., $\perp$) and his best deviation plan, respectively.
Hence, each $\vv \in {\calV}(o)$ is now a tuple $(v^\sfP, v^\sfA, v_*^\sfA)$.
(In \Cref{sc:backward-induction}, $v^\sfA$ and $v_*^\sfA$ are eventually forced to be the same, so there is no need to keep an additional dimension.)

The inducibility of a vector $\vv = (v^\sfP, v^\sfA, v_*^\sfA)$ is characterized by the following constraints.
First, we impose the same constraint as \Cref{eq:tV-omega-t-cons-v} on the first two dimensions of $\vv$, so that they capture the players' payoffs under the agent's truthful response.
In order for the third dimension $v_*^\sfA$ to capture the agent's maximum attainable value, we use a constraint similar to \Cref{eq:tV-omega-t-cons-ic}: 
\begin{align}
&&v_*^\sfA \ge
\sum_{\omega^\sfA} 
p_{h-1}(\omega^\sfA \given o)\ \max_{\wtomega^\sfA}\
\sum_{a^\sfA}\ 
\max_{\tda^\sfA}\
\sum_{s, \omega^\sfP, a^\sfP}\
p_{h-1}(s, \omega^\sfP \given o, \omega^\sfA) \cdot 
\hltb{\varpi ( \va  \given \omega^\sfP, \wtomega^\sfA )} \cdot 
\nonumber\\
&&
\left( r_h^\sfA \Big(s, a^\sfP, \tda^\sfA \right)  + 
% \phantom{\sum_{~}}\right. 
% \qquad\qquad\qquad\qquad\qquad\qquad\qquad\qquad\qquad\qquad\qquad
% \left.
% \sum_{o'} \phi_h \left(o' \given s, a^\sfP, \tda^\sfA \right) \cdot 
\hltb{{v'_*}^\sfA(s, \vomega, \wtomega^\sfA, \va, \tda^\sfA)}  \Big). 
\label{eq:max-attainable}
\end{align}
The remaining constraint is the same as \Cref{eq:tV-omega-t-cons-v-in-V}. 
% eq:tV-omega-t-cons-pi

All the non-linear constraints can be linearized the same way as the approach described in \Cref{sc:backward-induction}.
Hence, we can efficiently approximate ${\calV_h}(o)$ by examining the inducibility of points on a sufficiently fine-grained grid in $[0,H]^3$, which contains $\poly(H, 1/\epsilon)$ many points, and constructing the convex hull of these points. (Note that there is no need to ensure zero compromise on the agent's value as required in the proof of \Cref{lmm:compute-calV}. This is because $\delta$-IC is defined with respect to the agent's expected value at the beginning of the game instead of that at every time step. Hence, using points on a grid suffices the purpose of the approximation in this proof.)
The half-space representation of the convex hull can be computed efficiently given that it is in $\mathbb{R}^3$ \cite{chan1996optimal}.
Eventually, an optimal $\pi \in \wh{\Pi}_\delta$ corresponds to a solution to $\max_{\vv \in {\calV}_1(\varnothing)} v^\sfP$ subject to $v^\sfA \ge v_*^\sfA - \delta$, and we can use the same forward construction procedure in \Cref{sc:forward-construction} to compute $\pi_h(\cdot \given \sigma; \omega^\sfP, \wtomega^\sfA)$.

Note that \Cref{eq:max-attainable} only enforces $v_*^\sfA$ as an upper bound of the maximum attainable value, instead of the exact value.
This suffices for our purpose because any $(v^\sfP, v^\sfA, v_*^\sfA)$ in the feasible set $\calV_1(\varnothing) \cap \left\{\vv: v^\sfA \ge v_*^\sfA - \epsilon \right\}$ also implies the inclusion of $(v^\sfP, v^\sfA, \bar{v}_*^\sfA)$ in the same feasible set, where $\bar{v}_*^\sfA$ is the actual maximum attainable value induced by the policy that induces $(v^\sfP, v^\sfA, v_*^\sfA)$ according to our formulation.
\end{proof}

\thmlearning*

\begin{proof}
We run reward-free exploration to obtain a model $\whp$ with error bound $\delta/2$. 
This can be achieved w.h.p. in $\wt{\calO}(\zeta / \delta^2)$ episodes according to \Cref{lmm:Pi-hat-bounds}.
Next, we compute an $\delta$-optimal strategy $\pi \in \wh{\Pi}_{\delta}$ and use it in the remaining rounds. 
According to \Cref{lmm:opt-Pi-hat-eps}, this can be done in polynomial time.

By assumption, rewards are bounded in $[0,1]$ so the regrets are at most $1$ for both players in each of the exploration episodes.
In each of the remaining episodes, the agent's regret is as follows, where we pick arbitrary $\rho^* \in \argmax_{\rho} V_1^{\sfA,\pi,\rho}(\varnothing)$: 
\begin{align*}
V_1^{\sfA,\pi,\rho^*}(\varnothing) - V_1^{\sfA,\pi,\perp}(\varnothing) 
\; \le\; &
\underbrace{
\abs{\wh{V}_1^{\sfA,\pi,\rho^*}(\varnothing) - \wh{V}_1^{\sfA,\pi,\perp}(\varnothing)}
}_{\le \delta \text{ as } \pi \in \wh{\Pi}_{\delta}} \; +\; \\ 
&
\underbrace{
\abs{\wh{V}_1^{\sfA,\pi,\rho^*}(\varnothing) - V_1^{\sfA,\pi,\rho^*}(\varnothing)} 
\; +\; 
\abs{\wh{V}_1^{\sfA,\pi,\perp}(\varnothing) - V_1^{\sfA,\pi,\perp}(\varnothing)}
}_{\le \delta \text{ by \Cref{lmm:Pi-hat-bounds}}}
\; \le\; 
2 \delta.
\end{align*}
The principal's regret is: 
\begin{align*}
V^* - V_1^{\sfP,\pi,\perp}(\varnothing)
&\;=\;
\max_{\pi' \in \Pi_0} V_1^{\sfP,\pi',\perp}(\varnothing) - V_1^{\sfP,\pi,\perp}(\varnothing) \\
&\hspace{-8mm}
\underbrace{\;\le\; 
\max_{\pi' \in \wh{\Pi}_{\delta}} V_1^{\sfP,\pi',\perp}(\varnothing)}_{\text{as } \Pi_0 \subseteq \wh{\Pi}_{\delta}} - V_1^{\sfP,\pi,\perp}(\varnothing) 
\;\le\; 
\underbrace{\max_{\pi' \in \wh{\Pi}_{\delta}} \wh{V}_1^{\sfP,\pi',\perp}(\varnothing) - \wh{V}_1^{\sfP,\pi,\perp}(\varnothing)}_{\le \delta \text{ as $\pi$ is $\delta$-optimal}} \;+\; \delta
\;\le\; 
2\delta.
\end{align*}
The reason that $\Pi_0 \subseteq \wh{\Pi}_{\delta}$ is the following: 
Since $\whp$ ensures error bound $\delta/2$, we have $\abs{\wh{V}_1^{\sfA,\pi',\rho}(\varnothing) - V_1^{\sfA,\pi',\rho}(\varnothing)} \le \delta/2$ for all $\rho$.
By definition, $\pi' \in \Pi_0$ means that $V_1^{\sfA,\pi',\perp}(\varnothing) \ge V_1^{\sfA,\pi',\rho}(\varnothing)$. 
So, $\wh{V}_1^{\sfA,\pi',\perp}(\varnothing) \ge \wh{V}_1^{\sfA,\pi',\rho}(\varnothing) - \delta$ for all $\rho$; hence, $\pi' \in \wh{\Pi}_{\delta}$.

The above bounds then lead to a total regret of at most $\wt{\calO}(\zeta / \delta^2) + \calO(T\delta)$ for each player. 
Taking $\delta = (\zeta /T)^{1/3}$ gives the upper bound $\wt{\calO}(\zeta^{1/3} T^{2/3})$.
\end{proof}

\section{Complete Formulation of the Linear Constraints Satisfiability Problem}
\label{sec:lp}

The complete formulation of the linear constraint satisfiability problem in \Cref{sc:backward-induction}, resulting from the linearization of \Cref{eq:tV-omega-t-cons-v,eq:tV-omega-t-cons-ic}, is as follows,
where $\varpi$, $\zz = (z^\sfA, z^\sfP)$, and $y$ are the variables (highlighted in blue).

\begin{enumerate}[leftmargin=5mm,itemsep=1mm]
\item The value function constraint:
\begin{align*}
&\vv =
\sum_{s, \vomega, \va} 
p_{h-1}(s, \vomega \given o) \cdot \Big( \rr_h(s, \va) \cdot \hltb{\varpi({\va} \given \vomega)} +  
\hltb{\zz(s, \vomega, \omega^\sfA, \va, a^\sfA)}  \Big).
\end{align*}

\item An IC constraint for each $\omega^\sfA \in \Omega^\sfA$:
\begin{align*}
& 
\sum_{s, \omega^\sfP, \va}
p_{h-1}(s, \omega^\sfP \given o, \omega^\sfA) \cdot \Big( r_h^\sfA(s, \va)  \cdot \hltb{\varpi ( {\va} \given \vomega )} + 
\hltb{z^\sfA \left(s, \vomega, \omega^\sfA, \va, a^\sfA \right)}  \Big) 
\ge 
\sum_{ a^\sfA \in A^\sfA } \hltb{y\left( a^\sfA, \omega^\sfA, \wtomega^\sfA \right)}. 
\end{align*}
Moreover, for each tuple $(a^\sfA,\omega^\sfA,\wtomega^\sfA) \in A^\sfA \times \Omega^\sfA \times \Omega^\sfA$:
\begin{align*}
\hltb{y \left(a^\sfA, \omega^\sfA, \wtomega^\sfA \right)} \ge 
\sum_{s, \omega^\sfP, a^\sfP}\
p_{h-1}(s, \omega^\sfP \given o, \omega^\sfA) 
\Big( r_h^\sfA \left(s, a^\sfP, \tda^\sfA \right) \cdot \hltb{\varpi ( \va  \given \omega^\sfP, \wtomega^\sfA )}  + 
\hltb{z^\sfA\left(s, \vomega, \wtomega^\sfA, \va, \tda^\sfA \right)}  \Big).
\end{align*}

\item An onward value constraint for each tuple $(s, \vomega, \wtomega^\sfA, \va, \tda^\sfA) \in \bar{\Sigma}$:
\begin{align*}
&
\mathbf{H}\left(s, a^\sfP, \tda^\sfA \right) \cdot  \hltb{\zz(s, \vomega, \wtomega^\sfA, \va, \tda^\sfA)} 
\le \hltb{\varpi ( \va  \given \vomega )} \cdot \bb\left(s, a^\sfP, \tda^\sfA \right), 
\end{align*}
where for every $o \in O$, the matrix $\mathbf{H}(o)$ and vector $\bb(o)$ are given by the half-space representation of $\wh{\calV}_{h+1} (o)$, i.e.,  $\wh{\calV}_{h+1} (o) = \left\{\vv' \in \mathbb{R}: \mathbf{H}(o) \cdot \vv' \le \bb(o) \right\}$.

\item 
Additionally, we impose
\[
\hltb{\varpi(\cdot \given \vomega)} \in \Delta(A)
\]
for each $\vomega \in \Omega$
to ensure that $\varpi(\cdot \given \vomega)$ is a valid distribution over $A$.
\end{enumerate}

\section{Additional Discussion about Intractability without Hindsight Observability}
\label{sec:hardness}

The PSPACE-hardness can be seen by thinking of a POMDP as an instance of our problem where only the principal can make observations and perform actions to influence the environment (essentially, the agent can neither influence the principal nor the environment in this instance).

The PSPACE-hardness remains in the case of information design, where the principal observes the state directly but does not act, while the agent makes no observation but acts; as well as the case of mechanism design, where the agent observes the state directly but does not act, while the principal does not observe but acts.
This can be seen by considering zero-sum instances, where the principal's and the agent's rewards sum to zero. 

More specifically, consider for example the case of information design.
If the goal is to compute the principal's maximum attainable payoff, the PSPACE-hardness of the problem is immediate: Since the game is zero-sum, it is optimal for the principal to not send no signal (if signaling were to improve the principal's payoff, the agent would be better-off just ignoring the signals). 
Hence, computing the maximum attainable payoff of the principal in this case is equivalent to computing (the negative of) the agent's maximum attainable payoff, which amounts to solving a POMDP.

One may argue that while the above example demonstrates the hardness of determining the principal's maximum attainable payoff, computing the principal's optimal policy is actually trivial in the example (i.e., sending no signal is optimal). So it does not rule out the possibility of an efficient algorithm which, given any sequence, computes the signal distribution of an optimal policy, without computing the principal's payoff the policy yields.
It turns out that this is not possible, either. 

Consider a game where the agent can choose between two actions $a$ and $b$ in the first time step. Action $a$ leads to a process where the principal's rewards are zero for all state-action pairs.
Action $b$ leads to another process with payoffs $1-x$ for the principal and $x$ for the agent, where $x \in [0,1]$ depends on the principal's signaling strategy in this sub-process. 
For example, we can design this sub-process as a matching pennies game, where: nature flips a fair coin, the principal observes the outcome, and the agent must choose the same side of the coin to get a reward $1$ and otherwise he gets $-1$.
If the agent plays this matching pennies game on his own, his expected payoff is $0$. The principal can reveal her observation to help the agent to improve the payoff. And the principal can do so probabilistically, so that she can fine tune the agent's expected payoff $x$ to any desired value in $[0,1]$. 
To maximize the principal's payoff in the entire process requires finding an $x$ that is sufficiently high, so that the agent is incentivized to choose $b$ (otherwise, the principal only gets $0$); at the same time, we would like $x$ to be as low as possible to maximize the principal's payoff $1-x$.
This essentially requires knowing the agent's maximum attainable payoff in the sub-process following $a$, which is PSPACE-hard as we discussed above.

\end{document}